\documentclass[pra,showpacs,nobibnotes,nofootinbib,floatfix]{revtex4}
\pdfoutput=1

\usepackage{amsmath}
\usepackage{amssymb}
\usepackage{graphicx}
\usepackage{amsthm}
\usepackage{bbold}
\usepackage{xspace}
\usepackage{xfrac}  
\usepackage{float}

\usepackage{bbm}
\newcommand{\eins}{\ensuremath{\mathbbm 1}}

\usepackage{tikzexternal}
\usepackage[T1]{fontenc}
\usepackage[utf8x]{inputenc}

\usepackage{lmodern}

\DeclareMathOperator{\Tr}{Tr}

\newcommand{\R}{\ensuremath{\mathbb R}\xspace}
\newcommand{\C}{\ensuremath{\mathbb C}\xspace}
\newcommand{\csum}{\sideset{^c}{}\sum}

\theoremstyle{plain} 
  \newtheorem{thm}{Theorem}[section]
  \newtheorem{lem}[thm]{Lemma}
  
\theoremstyle{definition} 
  \newtheorem{defn}{{Definition}}[section]
  \newtheorem{rem}{{Remark}}[section]
  \newtheorem{exmpl}{{Example}}[section]

\begin{document}

\title{Entanglement and output entropy of the diagonal map}

\author{Meik Hellmund}
\affiliation{Mathematisches Institut, Universit{\"a}t Leipzig,
Johannisgasse 26, D-04103 Leipzig, Germany}
\email{Meik.Hellmund@math.uni-leipzig.de}

\begin{abstract}
We review some properties of the convex roof extension,
a construction used, e.g., in the definition of the entanglement of formation.
Especially we  consider 
 the use of symmetries of channels and 
states  for the construction of the convex roof.
As an application 
we study the entanglement entropy of the diagonal map for permutation
symmetric real $N=3$ states $\omega(z)$ and solve the case $z<0$ where $z$
is the non-diagonal entry in the density matrix.  We also report a
surprising result about the behaviour of the output entropy of the diagonal
map for arbitrary dimensions $N$; showing a bifurcation at $N=6$.   
\end{abstract}

\maketitle

\section{Introduction}
\label{intro}

Let  $\Phi\colon \omega\mapsto \omega'$ 
be a quantum channel or, somewhat more general, a trace-preserving 
positive map of (mixed) states $\omega\in \Omega$  from one quantum system
$\Omega$ to states $\omega'\in \Omega'$ from another system.  
We call
\begin{equation}
  \label{eq:1}
  E_\Phi(\omega) = \min_{\omega=\sum p_j \pi_j}\; \csum_j p_j\; S(\Phi(\pi_j))  
\end{equation}
{\it entanglement entropy} of the channel $\Phi$ or {\it $\Phi$-entanglement} for
short. Here the 
minimum is taken over all possible convex  
decompositions of the input state $\omega$ into pure states
\begin{equation}
  \label{eq:2}
  \omega =
\csum p_j \, \pi_j, \quad \pi_j \, \hbox{ pure, i.e., }
\pi_j = | \psi_j \rangle \langle \psi_j |
\end{equation}
and $S(\omega) =-\Tr \omega \log \omega$  is the von
Neumann entropy of the output states. We use the symbol { $ ^c\sum$} to
denote a convex sum, i.e., it implies  $p_j>0$ and $\sum p_j=1$.

The quantity \eqref{eq:1} appears in different places in quantum information
theory. For example,
\begin{enumerate}
\item The celebrated 
{\it entanglement of formation}\cite{BenFucSmo96} of a bipartite
  quantum system is the $\Phi$-entanglement of the partial trace
  $\Phi=\Tr_A$ with respect to one of the subsystems of the bipartite system.
\item 
The  theorem 
of Holevo, Schumacher, and Westmoreland\cite{schumacher97,holevo98}
 shows that  the one-shot or product state classical capacity
$\chi(\Phi)$ of a channel $\Phi$ can be obtained 
 by maximising the difference between 
output entropy and entanglement entropy (the so-called Holevo quantity) 
 over
all input density operators:
\begin{eqnarray}
  \label{eq:3}
\chi_\Phi^*(\omega) &=& S( \Phi( \omega )) - E_{\Phi}(\omega) \\ \nonumber
\chi_{\Phi} &=&  \max_{\omega} \;\chi_\Phi^*(\omega)
\end{eqnarray}
\item In \cite{BNU96} 
the optimization problem Eq.~\eqref{eq:1} was considered in
  connection with the quantum dynamical entropy of
  Connes-Narnhofer-Thirring\cite{CNT87}. In this framework
one considers a subalgebra $\mathcal{B}\subset \mathcal{A}$ of 
the algebra $\mathcal{A}$ 
of observables. The restriction of states to this subalgebra gives rise to a
channel $\Phi_{\mathcal{B},\mathcal{A}}$, 
the 
  dual of the inclusion map $\mathcal{B} \hookrightarrow \mathcal{A}$. 
  The difference $S(\Phi(\omega))- E_\Phi(\omega)$   is called  
{\it entropy of the subalgebra}; see also \cite{Benattibuch} for a   thorough
presentation.
\end{enumerate}
Closed formulas for the entanglement entropy, i.e.,
analytic solutions to the global optimization problem Eq.~\eqref{eq:1}
are very rare. 
They include  certain classes of highly symmetric states
\cite{TV00,vollbrecht-2000,caves08} and the celebrated entanglement of formation of a 
pair of qubits\cite{Woo97}.

Even earlier, Benatti, Narnhofer and Uhlmann 
\cite{BNU96,BUN99,BNU03} studied    the entanglement entropy of the diagonal map 
of a $3$-dimensional quantum system as an example for the entropy of a
subalgebra. 
The diagonal map (also called pinching channel) 
$\Phi_D$  sets all non-diagonal
elements of the input state $\omega$ to zero and corresponds to the choice of
a maximal abelian subalgebra $\mathcal{B}\subset \mathcal{A}$. 
Using a mixture of analytical and numerical methods, they found  explicit
results  for the entanglement entropy $E_{\Phi_D}$ (called $E_D$ in what
follows) of the
  diagonal map $\Phi_D$ 
applied to the 
one-dimensional family of  permutation symmetric $N=3$  
real input states 
\begin{equation}
  \label{eq:14}
  \omega(z) =\frac{1}{3}
  \begin{pmatrix}
    1 & z &z \\ z&1&z\\ z&z&1
  \end{pmatrix}.
\end{equation}

In this paper we present some remarks about the role of symmetries in the 
optimization problem \eqref{eq:1} based on the observations in \cite{TV00,vollbrecht-2000}. 
Using those insights we provide new results for the 
 entanglement entropy $E_D(z)$ of states of the
form \eqref{eq:14} for the case of negative values of the parameter $z$.

We also present a result about the output entropy of the diagonal map 
in arbitrary dimensions.

\section{Convex hulls and roof extensions}
\label{hull}

The state space $\Omega$  of a quantum mechanical system with an
$N$-dimensional Hilbert space $\mathcal{H}$ is a compact convex space of $N^2-1$ real
dimensions. 

A {\it (proper) face} $F$ of $\Omega$ is a non-empty subset $F\subsetneq
\Omega$ which is closed 
under convex compositions and 
decompositions, i.e., whenever $\omega = {}^c\sum_i \;p_i \omega_i$ and $\omega\in F$, then
$\omega_i\in F$. The (non-disjoint) union 
of all  faces $\bigcup F_i =\partial
\Omega$
constitutes the boundary 
of $\Omega$. There is a one-to-one correspondence between the faces of
$\Omega$ and linear
subspaces of  $\mathcal{H}$ with an $K^2-1$-dimensional face for every
$K$-dimensional subspace. The face consists of all the states $\omega$
with support in the corresponding subspace.    
Zero-dimensional faces correspond to pure states and
constitute the {\it extreme boundary} $\partial_e \Omega\subseteq \partial
\Omega$.

Let $f(\omega)$ be a real-valued function on $\Omega$. The convex hull
$f^\Cup$ of $f$ is 
the largest convex function not larger than $f$, i.e., 
for which $f^\Cup(\omega) \le f(\omega)\; \forall \omega\in \Omega$.   
The convex hull of a function is the solution of the global
  optimization problem 
  \begin{equation}
    \label{eq:28}
    f^\Cup(\omega)= \min_{{}^c\sum p_i \omega_i=\omega} \csum p_i f(\omega_i) 
  \end{equation}
where the minimum is taken over all convex decompositions of $\omega$. 
Carath\'{e}odory's
theorem asserts that we can restrict the search for optimal decompositions
 to decompositions of length up to
$l_{\max}=\dim\Omega+1$.

Let us now consider the case where  the function
$f$ is concave, such as the von Neumann entropy $S(\omega)$. 
 Obviously  we can then restrict the search for an
optimal decomposition to the extremal boundary, $\omega_i\in\partial_e\Omega$.   
It follows that the convex hull
 $f^\Cup$ depends only on the values of $f$ on $\partial_e\Omega$ and not on
 the behaviour of $f$ inside $\Omega$, as long as $f$ is everywhere concave.

Therefore we can consider an extension problem which ist closely related to
the global optimization problem \eqref{eq:28}:
Given a function $g(\pi)$ on $\partial_e\Omega$, i.e., on
the set of pure states, we ask  
for a canonical extension $g^\cup$ of $g$ to all of $\Omega$
defined as 
\begin{equation}
  \label{eq:18}
  g^\cup(\omega) = \min_{\omega=\sum p_j \pi_j}\; \csum_j p_j\; g(\pi_j)
\end{equation}.
This extension was called {\it convex roof extension} and intensively
studied in, e.g., \cite{Uh98,AUroof}. It is, in a sense, the extension which
is as linear as
possible while being everywhere convex.
\begin{defn}[roof extension]
  \label{thr:2}
  A function $G(\omega)$ is called a  {\it roof extension} of $g(\pi)$ if
  for every $\omega\in \Omega$ there is at least one extremal convex
  decomposition

  \begin{align}
    \label{eq:20}
 \omega&=\csum p_j \pi_j,\qquad \pi_j\in \partial_e \Omega \\
 \label{eq:21}
 \text{such that }\quad 
     G(\omega) &= \csum p_j g(\pi_j).     
  \end{align}
If this is the case, we call the decomposition \eqref{eq:20} {\it optimal}
with respect to $g$ or $g$-optimal.  
\end{defn}

\begin{figure}[h!tb]
  \centering
      \includegraphics[scale=1]{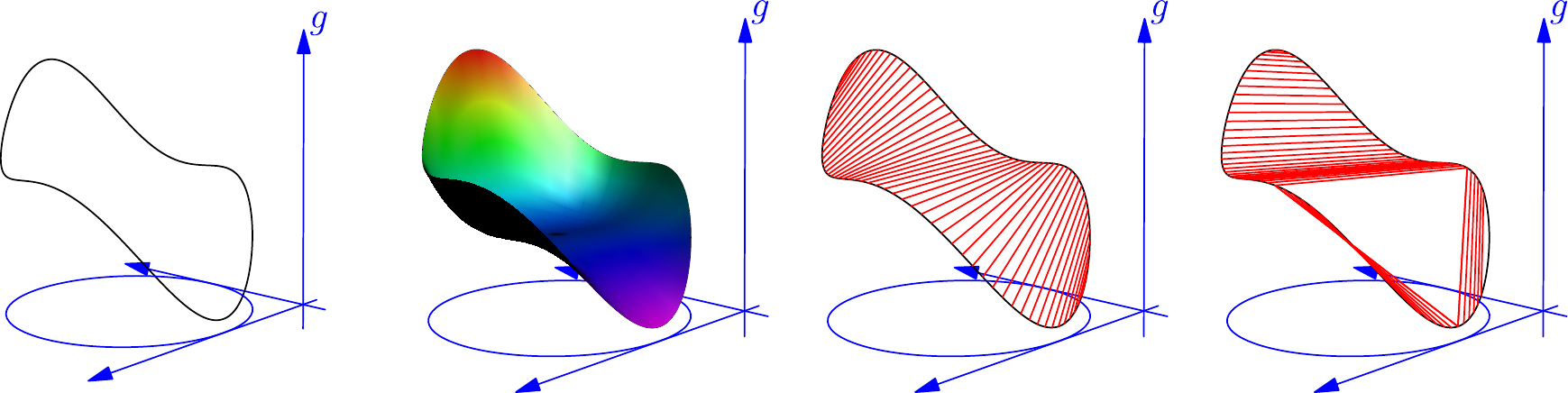}\\[.3cm]
 \hspace*{.9cm}   $g(\pi)$ \hspace*{2.cm} convex extension\hspace*{.8cm}
$\le$ \hspace*{1.1cm} roof extension  \hspace*{1.8cm}
\begin{minipage}[t]{5cm}
largest convex ext. = smallest roof ext. = convex roof   
\end{minipage}
  \caption{An illustration of the convex roof extension where $\Omega$ is a
    disc  and $\partial_e\Omega$  a circle. }
  \label{fig:roof}
\end{figure}

Fig.~\ref{fig:roof} may illustrate the concept and explain the name. In a
roof extension, the ground floor $\Omega$ is
covered by straight roof beams and plane tiles. Those beams and tiles rest with  their ends 
on  the the wall erected by $g(\pi)$. It is immediately clear from the
definition of convexity that every convex extension is pointwise majorized
by every roof extension. But is the largest convex extension a roof? 
The following theorem asserts that this is true at least when $g$ is
continuous and $\Omega$ compact:

\begin{thm}[\cite{BNU96,Uh98}]
  \label{thmRoof}
Let $g(\pi)$ be a continuous real-valued function on
the set of pure states $\partial_e\Omega$. 
There exists exactly one  function $g^\cup(\omega)$ on $\Omega$ which can be
characterized uniquely by each one of the  following four properties:
\begin{enumerate}
\item $g^\cup$ is the unique convex roof extension of $g$.
\item $g^\cup ( \omega )$ is the solution of the optimization problem
  \begin{equation}
    \label{eq:55}
    g^\cup ( \omega ) = \inf_{\omega = {}^c\sum p_j \, \pi_j }\; \csum p_j \, g( \pi_j ).
  \end{equation}
\item $g^\cup( \omega )$ is largest convex extension  as well as 
the smallest roof extension of $g$.
  \end{enumerate}
Furthermore, given $\omega \in \Omega$, the function $g^\cup$ is convex-linear
 on the convex hull of all
pure states $\pi$ appearing in optimal decompositions of $\omega$.

Therefore, $g^\cup$ provides a foliation of $\Omega$ into compact
leaves such that   a) each leaf is the convex hull of
some pure states and b) $g^\cup$ is convex-linear on each leaf.

\end{thm}

\begin{rem}
  \begin{enumerate}
  \item The theorem justifies to write ``min'' instead of ``inf'' in
Eqs.~\eqref{eq:1} and \eqref{eq:55}. 
\item If $g^\cup$ is not only linear but even constant on each leaf, it is called a 
flat roof.
\item 
Let $f(\omega)$ be a concave function on $\Omega$, e.g., $f(\omega) =
S(\Phi(\omega)) $.  
Then we denote by $f^\cup$ 
the convex roof extension of  $f\vert_{\partial_e\Omega}$.  
  \end{enumerate}
\end{rem}

\section{Symmetries and invariant states}
\label{symm}

The following lemma gives a simple bound for  $f^\cup$. 
Let $P: \Omega\rightarrow \Xi$ be an affine and
 surjective  map  $\omega  \mapsto \xi$. 
The space $\Xi$, as the image of $\Omega$ under an affine map, is convex and
compact,
but it need not be a quantum state space.
The map $P$ provides a
foliation of $\Omega$ into
leaves $\Omega = \bigcup_{\xi\in\Xi} L_\xi $ via 
\begin{equation}
  \label{eq:36}
L_\xi=\{\omega\, \vert\, \omega\in \Omega,  P(\omega)=\xi\}.  
\end{equation}
Since $P$ is affine, every leaf is generated by cutting $\Omega$ with some 
hyperplane 
and therefore, the leaves are convex, too. 
We define the function $\epsilon(\xi)$ on $\Xi$ as the minimum of $f$ on
the corresponding leaf 
\begin{equation}
  \label{eq:35}
  \epsilon(\xi) = \min_{L_\xi} f(\omega) =\min \{f(\omega)\,\vert\,P\omega=\xi\}.  
\end{equation}
\begin{lem}\label{lem:41}
  The convex hull  of the $\epsilon$ function, 
$ \epsilon^\Cup (\xi), $ provides a lower bound for the convex
  roof $f^\cup(\omega)$, i.e., 
  \begin{equation}
    \label{eq:37}
 f^\cup(\omega) \ge  \epsilon^\Cup(P\omega)
  \end{equation}
\end{lem}
\begin{proof}
  Let $\omega=\sum \lambda_i \omega_i$ be optimal for $f^\cup$, so
  $f^\cup(\omega) = \sum \lambda_i f(\omega_i)$. Let $\xi_i=P\omega_i$.
Then, due to linearity of $P$ we have $\xi=\sum\lambda_i\xi_i$ and
due to the definition of $\epsilon(\xi)$ we have  $f(\omega_i) \ge\epsilon(\xi_i) 
$. So,
\begin{equation}
  \label{eq:38}
  f^\cup(\omega) =\sum\lambda_if(\omega_i) \ge \sum \lambda_i \epsilon(\xi_i) 
\end{equation}
and from  eq.~\eqref{eq:28}, we have
\begin{equation}
  \label{eq:39}
  \sum \lambda_i \epsilon(\xi_i) \ge \epsilon^\Cup(\xi).
\end{equation}
\end{proof}
There are some cases where we can find states $\omega$ for which the 
inequality of the lemma can be sharpened to an equality. 

\begin{thm}
  \label{thr:1}
     Let $P=P^2$ be a
  linear and idempotent map of the state space onto itself 
  with an fixed point set $P\Omega=\Xi\subset \Omega$ of $P$-invariant
  states. 
Let $f(\omega) $ be $P$-invariant,i.e., 
  \begin{equation}
    \label{eq:51}
    f( P\omega) = f(\omega)\quad\forall\omega\in \Omega
  \end{equation}
Then
 for all $P$-invariant states $\omega_P\in P\Omega$ holds
  \begin{equation}
    \label{eq:15}
     f^\cup(\omega_P)=\epsilon^\Cup(\omega_P) = \epsilon^\cup(\omega_P) 
  \end{equation}
  and these states 
 have an optimal decomposition completely in $P\Omega$, i.e., into
  $P$-invariant states only. 
Furthermore,  for every state $\omega$ it holds that
  \begin{equation}
    \label{eq:16}
     f^\cup(\omega) \ge  f^\cup(P\omega)
  \end{equation}
\end{thm}
\begin{proof}
  Since $f$ is $P$-invariant, it is constant on every affine subspace
  $P^{-1}\omega$. Therefore, we have $\epsilon(\omega) = f(\omega)$ on
  $P\Omega$ and $\epsilon$ is concave. So, $\epsilon^\Cup(\omega) =
  \epsilon^\cup(\omega)$. 
Let $\omega={}^c\sum \lambda_j \omega_j$ be optimal for
$\epsilon^\cup$. Then,
\begin{equation}
  \label{eq:19}
  \epsilon^\cup(\omega) = \csum \lambda_j \epsilon(\omega_j) = \csum
  \lambda_j f(\omega_j) \ge f^\cup(\omega).
\end{equation}
Together with Lemma~\ref{lem:41} this proves \eqref{eq:15} and provides an
$f^\cup$-optimal decomposition lying in  $P\Omega$.  
\end{proof}

\begin{exmpl}
 Let $P$ be the projection to real states $P(\omega)=1/2
  (\omega+\omega^\top)$ and $f$ the output entropy of the diagonal map $S^D_\text{out}(\omega)=S(\Phi_D(\omega))$.
Then  real states have optimal decompositions into real
states only. Furthermore, Lemma~\ref{lem:41} asserts that all non-real states
have an entanglement entropy at least as large as their real projections
\begin{equation}
  \label{eq:47}
  E_D(\omega)=f^\cup(\omega)  \ge \epsilon^\cup(P \omega)= E_D(P\omega)\qquad
  \forall \omega\in \Omega
\end{equation}  
\end{exmpl}

A slightly different version was used in \cite{TV00} and 
worked out in \cite{vollbrecht-2000}: 
\begin{thm}[\cite{TV00,vollbrecht-2000}]
  \label{thr:3}
  Let $G$ be a symmetry goup of $f$
such that $f(\omega^g)=f(\omega)$  for all $\omega\in \Omega$ and $g\in G$.  
Let $P_G$ be the twirl map or group average corresponding to $G$, i.e., the 
idempotent projection to the subspace $P_G\Omega$ of 
$G$-invariant states 
  \begin{equation}
    \label{eq:43}
   \omega_P=  P_G\omega =\frac{1}{|G|} \sum_{g\in G} \omega^g
  \end{equation}
Then for all $G$-invariant states $\omega_P\in P_G\Omega$ holds 
\begin{equation}
  \label{eq:23}
  f^\cup(\omega_P) = \epsilon^\Cup(\omega_P).
\end{equation}
Furthermore, for every state $\omega\in \Omega$ it holds that 
\begin{equation}
  \label{eq:40}
  f^\cup(\omega) \ge f^\cup(P\omega)
\end{equation}
 States $\omega_P$ for which $\epsilon^\Cup(\omega_P)=\epsilon(\omega_P)$ 
have an optimal decomposition consisting of one complete orbit of $G$;
otherwise the optimal decomposition consists of several complete orbits. 
\end{thm}

\begin{proof}
  We assume that $\omega_P=\sum\lambda_i\omega_i$ is optimal for $
  \epsilon^\Cup(\omega_P)$. Let $\tilde\omega_i$  be  states which achieve the 
minimum in eq.~\eqref{eq:35} for
the $\omega_i$: $\epsilon(\omega_i) =f(\tilde\omega_i)$ and $\tilde\omega_i$ belongs
to the leaf $L_i$. So, $P\tilde\omega_i=\omega_i$ and therefore
\begin{equation}
  \label{eq:48}
  \omega_P=\sum_{i,g} \frac{\lambda_i}{|G|}\; \tilde\omega^g_i
\end{equation}
is a candidate decomposition for $f$. So,
\begin{equation}
  \label{eq:49}
  f^\cup(\omega_P) \le \sum_{i,g} \frac{\lambda_i}{|G|}\;f( \tilde\omega^g_i)
\end{equation}
With $f(\omega^g)=f(\omega)$ and $\sum_g\frac{1}{|G|}=1$, the right hand side
evaluates to
\begin{equation}
  \label{eq:50}
  f^\cup(\omega_P) \le \sum_i \lambda_i f(\tilde\omega_i) =\sum \lambda_i 
\epsilon(\omega_i) =\epsilon^\Cup(\omega)
\end{equation}
This, together with lemma \ref{lem:41}, proves the theorem and shows that 
decomposition \eqref{eq:48} is optimal for $f^\cup$. 
\end{proof}

Please note that the $G$-invariance of $f$ does not 
imply $P_G$-invariance of $f$. This is the main difference 
 between Theorem~\ref{thr:1} and Theorem~\ref{thr:3} for applications.

Only in the case where $f(\omega)$ is $P_G$-invariant $f(P\omega)=f(\omega)$ 
 (which implies
$G$-invariance $f(\omega^g)=f(\omega)$) we know that every $G$-invariant
state has an optimal decomposition consisting solely of $G$-invariant states.

\section{Output entropy of the diagonal map}
\label{output}

The diagonal map $\Phi_D$  maps $\Omega_N$, the state space of an
$N$-dimensional Hilbert space,  to the simplex 
$\Omega^\prime_N = \{x_1,x_2,...,x_N\},$ $ \;0\le x_i\le 1,$ $ \sum x_i=1$.
It corresponds to a complete von Neumann measurement. Its Kraus form is
\begin{equation}
  \label{eq:9}
  \Phi(\omega) = \sum_{i=1}^N \; P_i \omega P_i
\end{equation}
with $P_i=|i\rangle\langle i|$. The output entropy of this channel is
\begin{align}
  \label{eq:5}
  S_\text{out}^D(\omega) &= S(\Phi_D(\omega)) \\
  \label{eq:6}
  &= \sum_{i=1}^N \eta(x_i) 
\end{align}
with the usual abbreviation $\eta(x) = -x\log(x)$ for $x>0$ and $\eta(0)=0$.  
This function is not only concave but a concave roof,
 as was shown in \cite{AUroof}.

The minimal output entropy is zero and the maximal one is $\log N$.  

Things become more refined by restricting the channel onto a face of
$\Omega_N$. As an example we take the $(N-1)$-dimensional subspace
$\mathcal{H}_0$ which is orthogonal to the vector
$|\phi\rangle=N^{-1/2}\sum|j\rangle$.   It consists of vectors $\sum
a_j|j\rangle$ such that $\sum a_j=0$. $\mathcal{H}_0$ supports pure states
satisfying $\Phi_D(\pi)=N^{-1} \eins$ and so the maximal output entropy is 
$\log N$
again.

For the minimal output entropy we have a more complex result:
\begin{thm}\label{thm:1}
  Let $\Omega_0$ be the face of $\Omega_N$ consisting of states whose
  support is orthogonal
  to $\sum_{i=1}^N \vert i\rangle$. Let $S_\text{min}^D$ be the minimal
  output entropy of the diagonal map $\Phi_D$. Then  
we have:
\begin{itemize}
\item We have  $S_\text{min}^D(\Omega_0) =\log 2$ for 
For $N=2,3,\dots6$. This is
   achieved by the $N(N-1)/2$ pure input states $\pi_{jk}, \;j<k$,
  and only by these states. Here,
  $\pi_{jk}=|\phi_{jk}\rangle\langle\phi_{jk}|$ with
  $|\phi_{nm}\rangle=2^{-1/2} \left( |n\rangle-|m\rangle\right)$.
\item 
  For $N>6$ we have 
  \begin{equation}
    \label{eq:22}
    S_\text{min}^D(\Omega_0) = 
\log N -\left(1-\frac{2}{N}\right ) \log (N-1) 
  \end{equation}
 and $\lim_{N\rightarrow\infty}S^D_\text{min} =0.$
 The minimum  is achieved by the $N$ states
  $\pi_j=|\phi_{j}\rangle\langle\phi_{j}|$ where 
  \begin{equation}
    \label{eq:24}
\phi_1=    \left( \,(N-1)a,\,-a,\,-a\,,...,-a\,\right)
\qquad\text{ with }\qquad a= (N(N-1))^{-1/2}   
  \end{equation}
 and the other $\phi_i$ are obtained by permuting the components. 
\end{itemize}
\end{thm}
The proof of this theorem is found in the appendix.

\section{Entanglement entropy of the $N=3$ diagonal map for some subsets of states}
\label{ent}

\subsection{Geometry of the $N=3$ state space}
The space $\Omega_3$ of positive hermitean $3\times 3$ matrices with unit trace has 8 real
  dimensions. Its boundary consists of zero-dimensional faces (pure states) and
  three-dimensional faces (Bloch balls), the latter corresponding to
  two-dimensional subspaces of the Hilbert space $\mathcal{H}=\C^3$. 

We will use the notion $\psi=(a,b,c)$  to denote a 
one-dimensional subspace $[\C\psi]$ of
$\C^3$ as well as the corresponding point $\pi_\psi=\vert\psi\rangle\langle\psi\vert$ of $\Omega_3$. 
Here, $(a,b,c)$ is a generally
 unnormalized element of this one-dimensional
subspace. 
The set of states orthogonal to a given pure state $\psi$ form a Bloch
  ball which we denote by $B^\perp(\psi)$:
  \begin{equation}
    \label{eq:68}
    B^\perp(\psi) := \{\omega\,|\, \langle\psi|\omega|\psi\rangle=0\}
  \end{equation}
 and all non-trivial faces of $\Omega_3$ are obtained in this way: There is
 a Bloch ball opposite to each pure state and this gives a bijection between
 the 0- and 3-dimensional faces of $\Omega$.

More generally,  we can consider for every pure state $\psi$ 
 the foliation of $\Omega_3$ by parallel hyperplanes  $L_\psi(F)$ defined as 
 \begin{equation}
   \label{eq:69}
   L_\psi(F) := \{\omega\,|\;  \langle\psi|\omega|\psi\rangle=F\} 
 \end{equation}
where $F$ is the fidelity parameter.
The leaves are 7-dimensional in the generic case, but 
the highest leaf $L_\psi(1)=\pi_\psi$ consist of one pure state only and the
lowest leaf $L_\psi(0)=B^\perp(\psi)$ is the Bloch ball opposite to $\psi$. 
Furthermore,
every basis $\{\psi_1,\psi_2,\psi_3\}$ of three orthogonal pure states spans an
  equilateral triangle in $\Omega_3$. Every edge $\overline{\psi_i\psi_j}$ 
of this triangle is the
  diameter of the Bloch ball $B^\perp(\psi_k)$ orthogonal to the opposite
  vertex, see Fig.~\ref{fig:bbal}.
  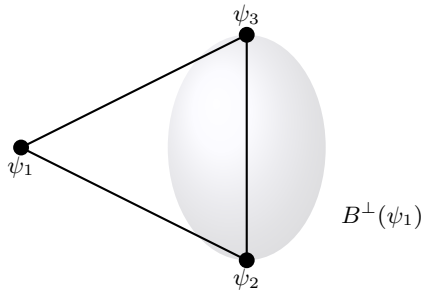
\begin{figure}[h!tb]
    \centering
  \begin{tikzpicture}[scale=1.5]
    \shade[ball color=blue!10!white,opacity=0.20] (0,0) ellipse (.7cm and 1cm);

    \draw[thick] (-2,0) -- (0.,1.) -- (0.,-1.) --(-2,0);
    \fill (-2,0) circle (2pt) node[below] {$\psi_1$};
    \fill (0.,1.0) circle (2pt) node[above] {$\psi_3$};
    \fill (0.,-1.) circle (2pt) node[below] {$\psi_2$};
    \node at (1.2,-.6) {$B^\perp(\psi_1)$};
  \end{tikzpicture}    
    \caption{Bloch ball face and opposite face. Note that the plane of the
      triangle and the ball have only a one-dimensional intersection, the
      diameter $\overline{\psi_2\psi_3}$}
    \label{fig:bbal}
  \end{figure}
All these triangles have the barycenter $\eins/3$ of $\Omega$ in common.

\subsection{States of lowest entanglement entropy}
\label{lowest}
 The triangle spanned by the computational basis 
$
\left(
    \begin{smallmatrix}
      1\\0\\0
    \end{smallmatrix}
\right), 
\left(
    \begin{smallmatrix}
      0\\1\\0
    \end{smallmatrix}
\right), 
\left(
    \begin{smallmatrix}
      0\\0\\1
    \end{smallmatrix}
\right) 
$ is the lowest leaf of the roof
of $E_D=S_D^\cup$, the leaf where $E_D=0$. Especially, $E(\eins/3)=0$. 
This triangle is the fixed point set of the diagonal channel $\Phi_D$.

\subsection{Some rank-2 states}

We can calculate $E_D$ for every Bloch ball which includes one of
the three states of the computational basis. Take, for example, $B(\psi_0,
\psi_{ab})$, the ball spanned by 
$\psi_0=(1,0,0)$ and some orthogonal state $\psi_{ab}=(0,a,b)$. 
This ball is the image of the unitary embedding 
$V\colon 
\left(
    \begin{smallmatrix}
      1\\0
    \end{smallmatrix}
\right)\mapsto
\left(
    \begin{smallmatrix}
      1\\0\\0
    \end{smallmatrix}
\right), 
\;
\left(
    \begin{smallmatrix}
      0\\1
    \end{smallmatrix}
\right) \mapsto
\left(
    \begin{smallmatrix}
      0\\a\\b
    \end{smallmatrix}
\right)
$ of a standard Bloch ball. This embedding can be used to
reduce the calculation of the convex roof to the $N=2$ case,
see chapter 6.2 in \cite{AUroof}. Using the known results 
for the $N=2$ diagonal map (see, e.g., \cite{BNU03}) we find
for states from this Bloch ball, i.e., states of the form 
\begin{equation}
  \label{eq:44}
  \omega=
  \begin{pmatrix}
    1-z& xa & xb \\ x^*a^* & z\, a^*a& z\, a^* b \\
    x^*b^* & z\, ab^* & z\, b^*b
  \end{pmatrix}
\end{equation}
with real $z$ and complex $x,a,b$ the entanglement entropy 
\begin{equation}
  \label{eq:45}
  E_D = \eta\left(\frac{1+\lambda}{2}\right) +
  \eta\left(\frac{1-\lambda}{2}\right) + z\,\eta(a^*a) + z\,\eta(b^*b)
\quad\text{where}\quad \lambda=\sqrt{1-4x^*x}.  
\end{equation}
\subsection{Real permutation invariant states}

The permutation group $G=S_3$ acts on $\Omega_3$ by permuting the
computational basis 
$
\left\{\left(
    \begin{smallmatrix}
      1\\0\\0
    \end{smallmatrix}
\right), 
\left(
    \begin{smallmatrix}
      0\\1\\0
    \end{smallmatrix}
\right), 
\left(
    \begin{smallmatrix}
      0\\0\\1
    \end{smallmatrix}
\right)\right\} 
$. The corresponding twirl acts on normalized pure states $\psi=(a,b,c)$ 
as 
\begin{equation}
  \label{eq:41}
    P_{G}( \pi_\psi) = 
\frac{1}{3}  \begin{pmatrix}
    1&z&z\\z&1&z\\z&z&1 
  \end{pmatrix} \quad\text{where}\quad 2z=ab^*+a^*b+ac^*+a^*c+bc^*+b^*c.
\end{equation}
Let us  denote the $S_3$-invariant state on the right hand side as $\omega(z)$. 
In what follows, 
we restrict our considerations to real states $\omega=\omega^\top$. 
Then the real parameter $z$ can take values in the range
\begin{equation}
  \label{eq:46}
  -\frac{1}{2}\le z\le 1.
\end{equation}
Another often used parametrization for these states uses the fidelity
 with respect to the state 
$\omega(1)=\pi(1,1,1)$. We have 
\begin{equation}
  \label{eq:54}
  F=\langle\psi_{(1,1,1)}\vert\omega\vert\psi_{( 1,1,1)}\rangle = \frac{1}{3}(2z+1).
\end{equation}
The state $\omega(z)$ is of rank three except at the boundaries of the $z$
range:
\begin{itemize}
\item For $z=1$ we have  a pure state 
  \begin{equation}
    \label{eq:57}
\omega(1)=\pi(1,1,1)
  \end{equation}
where we
  use $\pi(a,b,c)$ as shorthand for  
  the pure state $\pi_\psi=\vert\psi\rangle\langle\psi\vert$ with
  $\psi\sim (a,b,c)$. So its entanglement entropy equals its output entropy 
and we have 
\begin{equation}
  \label{eq:59}
  E_D(1) = \log(3).
\end{equation}
\item For $z=-\sfrac{1}{2}$ we have a rank-2 state  
  \begin{equation}
    \label{eq:58}
\omega( -\sfrac{1}{2} ) = \frac{1}{2}\left[
  \pi(-1,0,1)+ \pi(-1,2,-1)\right].      
  \end{equation}
 This state belongs to the face of $\Omega$ considered in
 Section~\ref{output}, so its entanglement entropy can't be smaller than
 $\log(2)$. It is easy to see that this value can indeed be achieved by the optimal
decomposition
\begin{equation}
  \label{eq:61}
  \omega\left(-\sfrac{1}{2}\right)  
 = \frac{1}{3}\left[\pi(1,0,-1) +
  \pi(0,-1,1)+\pi(-1,1,0)\right],
\end{equation}
and therefore 
\begin{equation}
  \label{eq:60}
  E_D\left(-\sfrac{1}{2}\right) = \log(2)
\end{equation}
\item Let's also mention that for $z=0$ we have the maximally mixed state which belongs to the
  lowest leaf of Section~\ref{lowest}. So an optimal decomposition is
  \begin{align}
    \label{eq:67}
    \omega(0) &= \frac{1}{3}\left[\pi(1,0,0) +
  \pi(0,1,0)+\pi(0,0,1)\right], \\
\label{eq:70}
E_D(0) &= 0.
  \end{align}
\end{itemize}

Applying Theorem~\ref{thr:1} using the projection
$P\omega=\frac{1}{2}(\omega+\omega^\top)$ we see that the states $\omega(z)$ have optimal
decompositions including only real states. Furthermore, for an arbitrary
state we have
\begin{equation}
  \label{eq:52}
  E_D(\omega) \ge E_D\left( \frac{\omega+\omega^\top}{2}\right).
\end{equation}
Applying now Theorem~\ref{thr:3} with the projection $P_{G}$  to the space
of real states we learn that 
\begin{equation}
  \label{eq:53}
  E_D(z) =  \epsilon^\Cup(z) \quad\text{where}\quad
  \epsilon(z)=\min_{P_{G}(\pi(a,b,c))=\omega(z)}\; 
S_D\left(\pi(a,b,c)\right),\quad \text{with }z,a,b,c\text{ real}
\end{equation}
The minimization in Eq.~\eqref{eq:53} is one-dimensional  since 
the three real parameters are constrained by $a^2+b^2+c^2=1,$ $ ab+bc+ac=z$. 
A useful parametrization of this constraint is \cite{BNU96} 
\begin{equation}
  \label{eq:56}
  3\,(a,b,c) = (\,\alpha+2\beta \cos \theta,\; \alpha-2\beta\cos(\theta-\pi/3),\;\alpha-2\beta\cos(\theta+\pi/3)\,) 
\quad\text{where}\quad \alpha=\sqrt{2z+1},\; \beta=\sqrt{1-z}
\end{equation}
Numerical search for the minimum in  Eq.~\eqref{eq:53} shows that the
minimum is reached at $\theta=0$ for all $z> -0.4150234$. For smaller values
of $z$, $\theta_\text{min}$ increase up to $\theta_\text{min}=\sfrac{\pi}{6}$ 
at $z= -\sfrac{1}{2}$. A thorough analysis of the function $\epsilon(z)$
obtained by this minimization shows that it is not everywhere convex. 
In the region $z\ge 0$ we re-obtained the result of \cite{BNU03}: the convex
hull $\epsilon^\Cup$  is obtained by replacing $\epsilon(z)$ in the region
$\sfrac{5}{6}<z\le 1$ with a linear piece. 

In the negative-$z$ region our results differ from \cite{BNU03}, who claimed
that  $\epsilon(z)$ is convex there. We find that the convex hull is
obtained by replacing $\epsilon(z)$ in the region between
$z=-\sfrac{1}{2}$ and $z=z^*=-0.4079496711$ with a linear piece, see
Fig.~\ref{fig:eee}.  

\begin{figure}[h!tb]
  \centering
  \begin{tikzpicture}
    \begin{axis}[xlabel=$z$, ylabel=$E_D(z)$, width=14cm,
      height=8cm,xmin=-.5, xmax=1., ymin=-.1, ymax=1.5]
      \addplot gnuplot [id=plot1,domain=-.40794967:.83333,samples=50, no marks,thick,color=blue] {
(-2*(-sqrt(1 - x) + sqrt(1 + 2*x))**2*
      log(((-sqrt(1 - x) + sqrt(1 + 2*x))**2)/9.))/9. - 
   ((2*sqrt(1 - x) + sqrt(1 + 2*x))**2 *
      log(((2*sqrt(1 - x) + sqrt(1 + 2*x))**2) /9.))/9.
      };
      \addplot[color=red,thick,mark=*] coordinates {
        (1.,1.09861)
        (.8333,0.867563)
        };
      \addplot[color=red,thick,mark=*] coordinates {
        (-.5,0.693147)
        (-.40794967,0.470016)
        };

    \end{axis}      

    \begin{scope}[xshift=5cm,yshift=2cm,scale=2.2]
      \draw[white,fill=gray!30] (0,0,0) -- (0,0.2,1) -- (0,.7,0.5) -- cycle;
      \draw[thick] (0,0,0) -- (0,.35,.5);
      \draw[thick] (0,0.2,1) -- (0,.35,.5);
      \draw[thick] (0,.7,.5) -- (0,.35,.5);
       \draw[fill, blue] (0,0,0) circle (.03);
       \draw[fill, blue] (0,0.2,1) circle (.03);
       \draw[fill, blue] (0,0.7,.5) circle (.03);
       \draw[fill, red] (0.,0.35,.5) circle (.04);

    \end{scope}

    \begin{scope}[xshift=10cm,yshift=5cm,scale=1.5]
      \draw[white,fill=gray!30] (0,0,0) -- (0,0,1) -- (0,.7,0.5) -- cycle;
      \draw[thick] (0,0,0) -- (1,.35,.5);
      \draw[thick] (0,0,1) -- (1,.35,.5);
      \draw[thick] (0,.7,.5) -- (1,.35,.5);
      \draw[thick] (1.5,.35,.5) -- (1,.35,.5);
       \draw[fill, blue] (0,0,0) circle (.05);
       \draw[fill, blue] (0,0,1) circle (.05);
       \draw[fill, blue] (0,0.7,.5) circle (.05);
       \draw[fill, blue] (1.5,0.35,.5) circle (.05);
       \draw[fill, red] (1.,0.35,.5) circle (.07);
  \node at (1.3,-.02,0) {{\scriptsize $z=1$ } }; 
  \node at (-.3,-.5,0) {{\scriptsize $z=5/6$ } }; 

    \end{scope}

    \begin{scope}[xshift=3.6cm,yshift=5cm,scale=1.5]
      \draw[white,fill=gray!30] (0,0,0) -- (0,0,1) -- (0,.7,0.5) -- cycle;
      \draw[white,fill=gray!30] (-1.5,0,0.5) -- (-1.5,0.7,0) -- (-1.5,.7,1.) -- cycle;
      \draw[thick] (0,0,0) -- (-1,.35,.5);
      \draw[thick] (0,0,1) -- (-1,.35,.5);
      \draw[thick] (0,.7,.5) -- (-1,.35,.5);
      \draw[thick] (-1.5,.0,.5) -- (-1,.35,.5);
      \draw[thick] (-1.5,.7,.0) -- (-1,.35,.5);
      \draw[thick] (-1.5,.7,1.0) -- (-1,.35,.5);
       \draw[fill, blue] (0,0,0) circle (.05);
       \draw[fill, blue] (0,0,1) circle (.05);
       \draw[fill, blue] (0,0.7,.5) circle (.05);
       \draw[fill, red] (-1.,0.35,.5) circle (.07);
       \draw[fill, blue] (-1.5,0.,.5) circle (.05);
       \draw[fill, blue] (-1.5,0.7,0) circle (.05);
       \draw[fill, blue] (-1.5,0.7,1.0) circle (.05);
  \node at (-1.7,-.4,0) {{\scriptsize $z=-0.5$ } }; 
  \node at (-.2,-.55,0) {{\scriptsize $z=z^*$ } }; 
    \end{scope}

  \end{tikzpicture}
  \caption{Entanglement entropy and optimal decompositions. Regions where
    the convex hull construction leads to a linear behaviour are drawn in
    red. The drawings indicate the shape of optimal decompositions of the
    red state into pure (blue) states in the
    three regions $z< z^*,\; z^*\le z< \sfrac{5}{6},\; \sfrac{5}{6}< z$, resp.  Here,  gray surfaces
    indicate hyperplanes of states which project under $P_G$  to 
   the same value of  $z$. The corners of the triangles form a complete
   orbit under the permutation group $G=S_3.$}
  \label{fig:eee}
\end{figure}
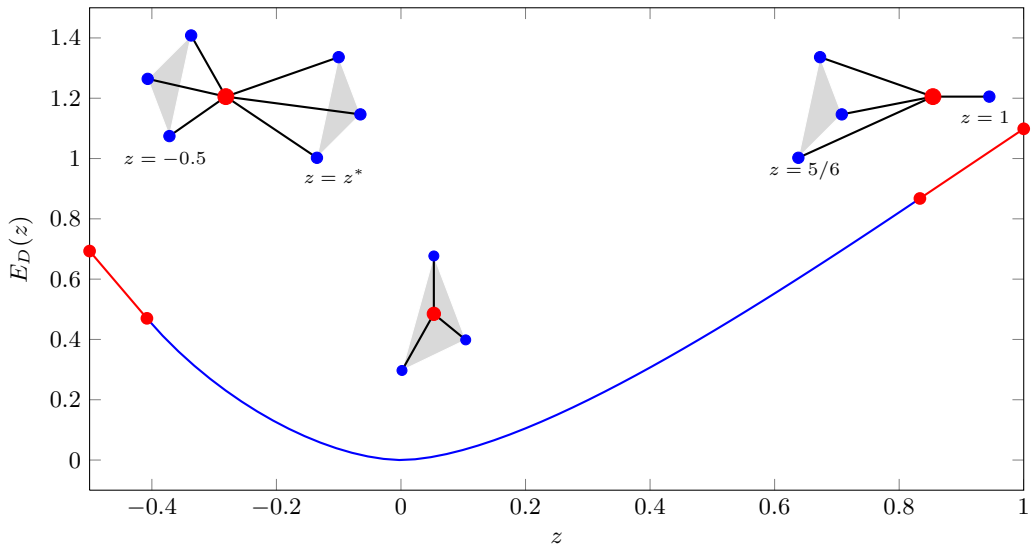

Interestingly, everywhere in the region $z^*\le z\le \sfrac{5}{6}$ where
$\epsilon(z)=\epsilon^\Cup(z)$ the minimum is obtained for states 
with $\theta=0$. So the optimal decompositions in this region have the form 
\begin{equation}
  \label{eq:72}
  \omega = \frac{1}{3}\left[
\pi(a,b,b)+\pi(b,b,a)+\pi(b,a,b)\right],
\end{equation}
corresponding to a short orbit of $S_3$ of length 3 only. In the region 
$-\sfrac{1}{2} <z<z^*$ the optimal decomposition has length 6 and is a
mixture of two such short orbits resulting in a large region in state space
where the entanglement entropy is an affine function.

With 
\begin{equation}
  \label{eq:66}
\mathcal{S}(z) = 2 \eta\left(\frac{(\alpha-\beta)^2}{9}\right)
   + \eta\left(\frac{(\alpha+2\beta)^2}{9}\right)
\quad\text{where}\quad \alpha=\sqrt{2z+1},\; \beta=\sqrt{1-z}
\end{equation}
and $z^*=-0.4079496711$, $\mathcal{S}(\sfrac{5}{6}) =
\log(3)-\sfrac{1}{3}\log(2)$ the final result for the entanglement entropy
is therefore
\begin{equation}
  \label{eq:71}
  E_D(z) =
  \begin{cases}
    p \log(2) + (1-p) \mathcal{S}(z^*) &\text{for } -\sfrac{1}{2}\le z< z^*
\text{ with }  p = \frac{z^*-z}{z^*+\sfrac{1}{2}}
 \\
 \quad   \mathcal{S}(z) & \text{for } z^*\le z\le \sfrac{5}{6}\\
    p [\log(3)-\sfrac{1}{3}\log(2)] +(1-p) \log(3) &
\text{for } \sfrac{5}{6}<z\le 1 \text{ with }  p = \frac{1-z}{1-\sfrac{5}{6}}
  \end{cases}
\end{equation}


\begin{acknowledgments}
  I would like to thank Armin Uhlmann for encouragement and many useful
 explanations and  discussions. 
\end{acknowledgments}

\appendix
\section{Proof of theorem \ref{thm:1} }
\begin{enumerate}
  \item We have $S^D_\text{out}(\omega)=S^D_\text{out}(\omega^\top) =
    S^D_\text{out}(\frac{1}{2}(\omega+\omega^\top))$. Therefore we can
    restrict our search for the minimum to the subspace $\Omega_\R
    =\{\omega\;\vert\;\omega=\omega^\top\}$ of real states.
  \item 
The minimal output entropy is attained by  pure  states  since
    $S_\text{out}^D$ 
    is concave and $\Omega_\R$ is convex. 
  \item The case $N=2$ is trivial. There is only one pure real 
state $\phi_{12}=\phi_1
    = 2^{-1/2} (1,-1)$ in $\mathcal{H}_0$ 
with $S_\text{out}^D=\log 2$. So we now assume $N\ge 3$. 
  \item The pure real states in $\mathcal{H}_0$ have the form
    $\pi=|\psi\rangle\langle \psi|$ with
    \begin{equation}
      \label{eq:11}
      |\psi\rangle =\sum a_j|j\rangle, \qquad \sum a_j=0,\qquad \sum a_j^2=1
\qquad\text{which implies}\qquad a_i\in (-1,1).
    \end{equation}
    So we use Lagrange's multiplier method to find the minimum of
    \begin{equation}
      \label{eq:12}
       S_\text{out}^D(\pi) = -\sum a_j^2\log a_j^2
    \end{equation}
for all $a_i$ satisfying eq.~\eqref{eq:11}. 
The equations to solve read 
\begin{align}
  \label{eq:25}
  a_i \log a_i^2 + a_i &= \lambda+ a_i \mu' \\
  \label{eq:26}
\text{ or: }\quad\qquad a_i \log a_i^2 &= \lambda + a_i \mu\qquad 
\end{align}
where $\lambda,\mu'$ denote Lagrange multipliers for the constraints 
Eq.~\eqref{eq:11}
and $\mu=\mu'-1$. 
Multiplying eq.~\eqref{eq:26} by $a_j$ and summing over $j$ yields
\begin{equation}
  \label{eq:27}
  S_\text{out}^D=-\mu \qquad\text{and therefore,}\qquad \mu<0.
\end{equation}
Not all solutions of eq.~\eqref{eq:26} have minimal output entropy but all
states of minimal entropy must be solutions of eq.~\eqref{eq:26}. So can
find the minimum by classifying all solutions and comparing their entropy. 
Let us consider different cases:
\begin{enumerate}
\item $\lambda=0$, so the solutions of eq.~\eqref{eq:26} are $a_i\in\{0,\pm
  \exp(\mu/2)\}$. Let $m$ instances of the $a_i$ be nonzero. 
Then their modulus must
  be $m^{-1/2}$ for $\sum a_j^2=1$ and $S_\text{out}^D=-\mu=\log m$. 
 Since $m$ must be even for $\sum
  a_j=0$, the minimum value for $S_D$ is achieved for $m=2$. So one
  candidate for the minimum of $S_\text{out}^D$ is 
  \begin{equation}
    \label{eq:29}
    S_\text{out}^D=\log 2,\qquad 
|\phi_{jk}\rangle=2^{-1/2} \left( |j\rangle-|k\rangle\right)
  \end{equation}
\item $\lambda\ne 0$. Then all the $a_i$ have to be non-zero.
The transcendental equation $\lambda+ x\mu =x \log x^2$ 
can be rewritten as
\begin{equation}
  \label{eq:62}
  \frac{\lambda}{2} e^{-\frac{\mu}{2}} = \frac{\lambda}{2|x|}
  e^{\frac{\lambda}{2x}} = \pm we^w \quad\text{with}\quad w= \lambda/2x 
\end{equation}
The inverse of the function $f(w)= w e^w$ is  the Lambert $W$ function
$W(z)$\cite{lambert1}, defined via 
\begin{equation}
  \label{eq:7}
  z = W(x)\; e^{W(z)}.
\end{equation}
As an inverse of a non-injective function it has multiple branches, two of
which are real and  denoted as $W_{-1}$ and $W_0$. 
It follows that 
 $\lambda+\mu x= x \log x^2$  has no more than three  real 
solutions which can be
 expressed as 
\begin{align}
  \label{eq:63}
  x_1 &= \frac{\lambda}{2 W_0(z)} \quad\text{where}\quad
  z=\frac{|\lambda|}{2}e^{-\mu/2} \\
  \label{eq:64}
  x_2 &= \frac{\lambda}{2 W_0(-z)} \\
  \label{eq:65}
  x_3 &= \frac{\lambda}{2 W_{-1}(-z)}
\end{align}
Since $\lambda\ne0$ we have $z>0$. Then a solution $x_1$ does always exist
and 
the solutions $x_2$ and $x_3$ exist only if $z\le 1/e$. They are equal for
$z=1/e$.  

\begin{enumerate}
\item Let us assume that only two of the values, say $x_1$ and $x_2$ are
  used in the state. So we have 
  \begin{equation}
    \label{eq:30}
    n x_1+m x_2=0,\quad nx_1^2+m x_2^2=1, \quad n+m=N,\quad n,m\ge 1
  \end{equation}
resulting in $x_1^2=\frac{m}{nN},\;\;x_2^2=\frac{n}{mN}$ and so
\begin{equation}
  \label{eq:31}
  S(N,n) = \log N -\left(1-\frac{2n}{N}\right)\log\left(\frac{N}{n}-1 \right)
\end{equation}
Now this expression is concave in $n$: 
\begin{equation}
  \label{eq:32}
  \frac{\partial^2 S(N,m)}{\partial n^2}= - \frac{N^2}{n^2(N-n)^2} < 0
  \quad\text{for}\quad 0< n< N 
\end{equation}
and therefore takes for fixed $N$  its minimum at the edges of the allowed
$n$-range, $n=1$ or, equivalently, $n=N-1$. 

So 
 the second  possibility for a minimum of $S_\text{out}^D$ is 
\begin{equation}
  \label{eq:33}
  S_\text{out}^D = \log N -\left(1-\frac{2}{N}\right)\log\left(N-1 \right),
\qquad
\{a_j\} =(a,a,...,a,(1-N)a)
\end{equation}
\item The last possibility is that all three roots $x_1,x_2,x_3$ occur 
among the $a_i$. 
Consider the function 
\begin{equation}
  \label{eq:34}
  F(\lambda,\mu) = x_1^2+x_2^2+x_3^2 
\end{equation}
where the $x_i$  are the    three solutions  of $\lambda +\mu x=x\log x^2$.
Using Eqs.~(\ref{eq:7},\ref{eq:63},\ref{eq:64},\ref{eq:65}) 
we find
\begin{align}
  \label{eq:8}
  F(\lambda,\mu) &=  e^\mu\left( e^{2W_0(z)} + e^{2W_0(-z)} + e^{2W_{-1}(-z)}
\right) \\
\label{eq:10}
 &= e^\mu G(\lambda,\mu) \\
 \label{eq:13}
 \text{where}\qquad G(\lambda,\mu) &= e^{2W_0(z)} + e^{2W_0(-z)} +
 e^{2W_{-1}(-z)} 
\end{align}
\begin{figure}[h!tb]
  \centering
  \includegraphics[scale=.5]{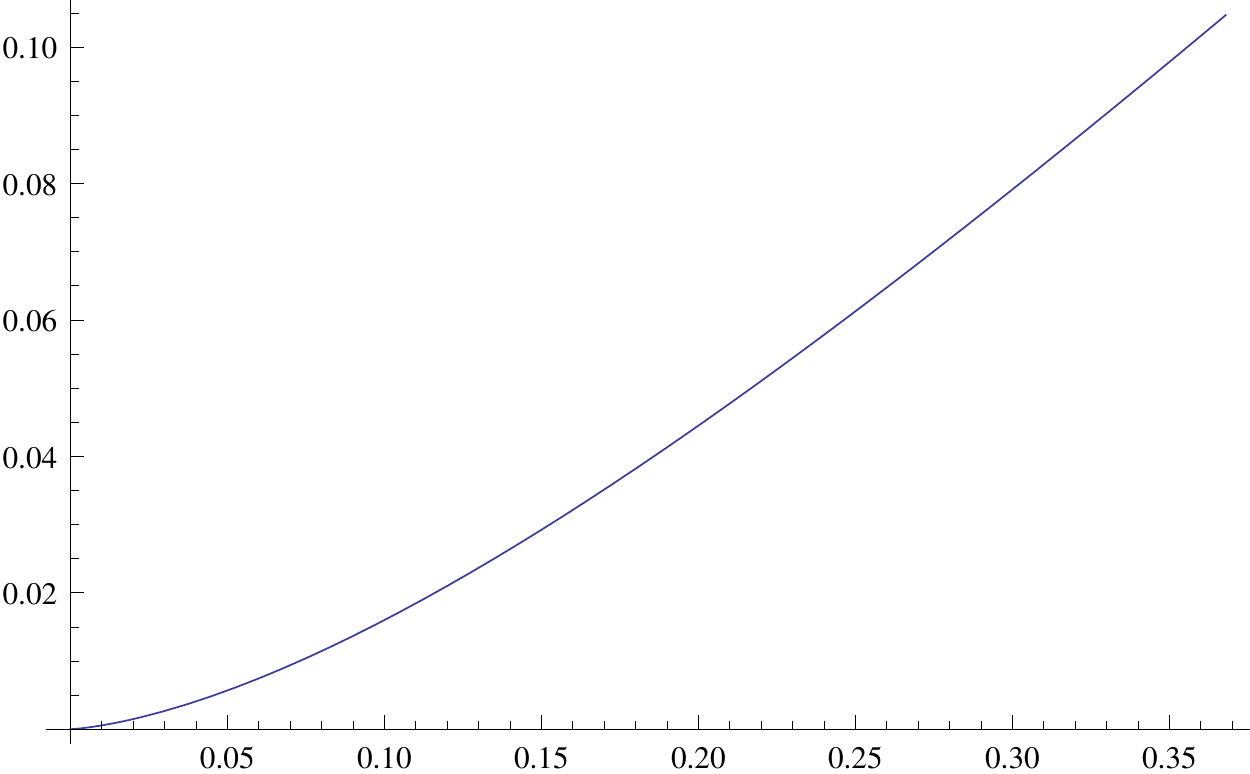}
  \caption{Plot of $G'(z)$ over $z$}
  \label{fig:gp}
\end{figure}
Now we have $G(0)=2$ and $G'(z) > 0 $ for $z\in (0,1/e]$ (see
Fig.~\ref{fig:gp}), therefore $G(z)> 2$ for $z\in (0,1/e]$.  
 Since we need $F\le1 $ for a normalized state vector, this implies
$e^\mu<\frac{1}{2}$, $\mu<-\log 2$, $ S^D_\text{out}=-\mu> \log 2$ 
and therefore any such solution has larger output entropy than the state 
given by Eq.~\eqref{eq:29}.
\end{enumerate}
\end{enumerate}

\item 
The only thing left to do is to compare the two candidates for a minimum,  
Eqs.~\eqref{eq:33} and \eqref{eq:29}. It is easy to see, that 
candidate \eqref{eq:29}
\begin{figure}[h!tb]
  \centering
  \includegraphics[scale=.6]{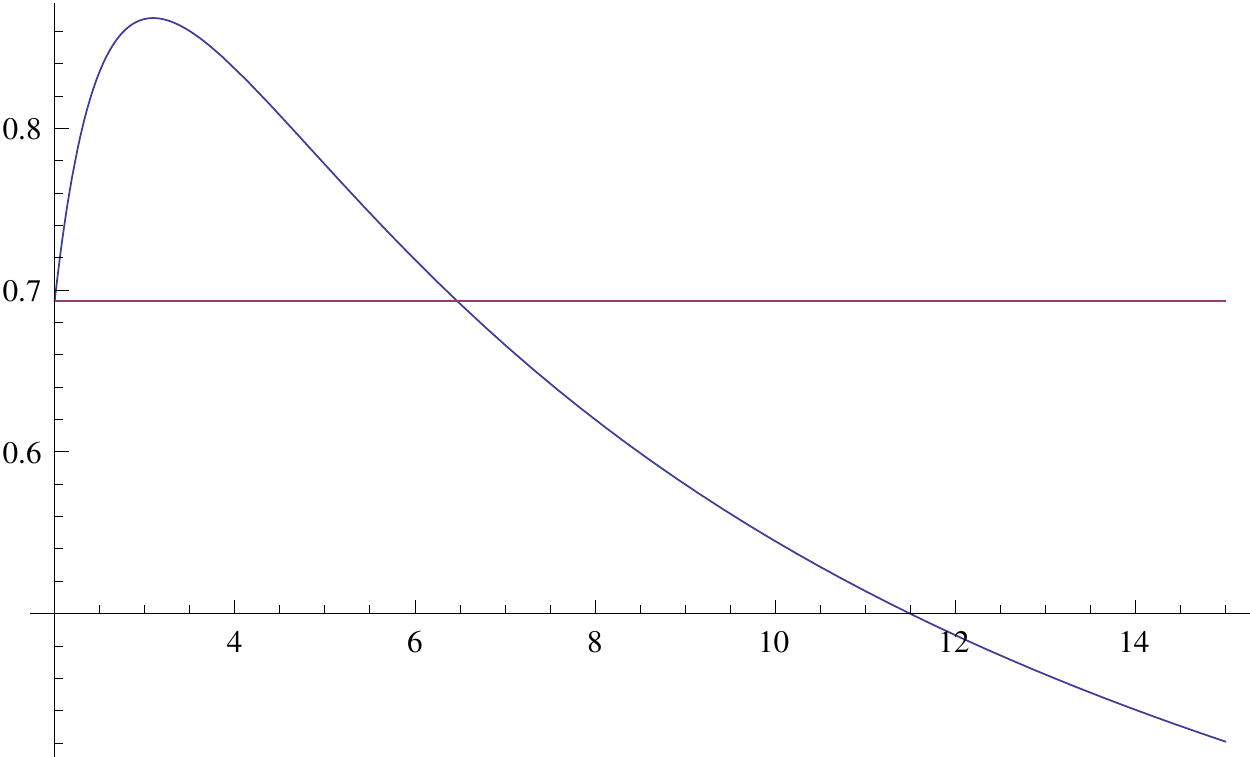}  
  \caption{The output entropies $\log 2$ and Eq.~\eqref{eq:33} plotted over $N$}
  \label{fig:outent}
\end{figure}
wins for $N\le 6$ and candidate \eqref{eq:33} 
wins for all $N>6$, see Fig.~\ref{fig:outent}. 
\hfill$\square$

\end{enumerate}

\bibliographystyle{apsrev}
\bibliography{dach}

\end{document}